\documentclass[onefignum, onetabnum]{siamart171218}

\usepackage{lipsum}
\usepackage{amsfonts}
\usepackage{graphicx}
\usepackage{epstopdf}
\usepackage{algorithmic}
\usepackage{tikz}
\usepackage{subcaption} 
\usepackage{pc_math}
\usepackage{amssymb}
\usepackage{bbm}
\usepackage[color=gray!60]{todonotes}
\usepackage[utf8]{inputenc}
\graphicspath{{fig/}}

\tikzstyle{major}=[circle, fill=blue!50, minimum size=10pt, line width=0mm, inner sep=0pt, draw=black]
\tikzstyle{minor} = [major, fill=orange!80, draw=black]
\tikzstyle{focus} = [minor, line width=.2mm, inner sep=1pt, draw=black]
\tikzstyle{edge} = [draw,line width=.3mm, -, gray!120]
\tikzstyle{gone} = [edge,  densely dotted]
\tikzstyle{arrow} = [draw,line width=.2mm, ->, gray!120]

\tikzstyle{weight} = [font=\small]

\ifpdf
  \DeclareGraphicsExtensions{.eps,.pdf,.png,.jpg}
\else
  \DeclareGraphicsExtensions{.eps}
\fi

\newsiamremark{remark}{Remark}
\newsiamremark{hypothesis}{Hypothesis}
\crefname{hypothesis}{Hypothesis}{Hypotheses}
\newsiamthm{claim}{Claim}
\newsiamthm{dfn}{Definition}

\headers{Adaptive Voter Models}{P. S. Chodrow, P. J. Mucha}

\title{Local Symmetry and Global Structure in Adaptive Voter Models\thanks{Submitted to the editors \today.
\funding{PSC was supported by an NSF Graduate Research Fellowship under Grant Number 1122374. PJM was supported by the Eunice Kennedy Shriver National Institute of Child Health \& Human Development of the National Institutes of Health under Award Number R01HD075712. The content is solely the responsibility of the authors and does not necessarily represent the official views of the National Science Foundation or National Institutes of Health. 
}}}

\author{Philip S. Chodrow\thanks{Operations Research Center, Massachusetts Institute of Technology, Cambridge, MA, 02139 and Laboratory for Information and Decision Systems, Massachusetts Institute of Technology, Cambridge, MA, 02139
  (\email{pchodrow@mit.edu}).}
\and Peter J. Mucha\thanks{Carolina Center for Interdisciplinary Applied Mathematics, Department of Mathematics, University of North Carolina, Chapel Hill, NC 27599-3250, and Department of Applied Physical Sciences, University of North Carolina, Chapel Hill, NC 27599-3050. 
  (\email{mucha@unc.edu}).}}

\usepackage{amsopn}

\usepackage{cite}

\usepackage{color}
\usepackage{xcolor}
\colorlet{comment_purple}{blue!50!red}

\usepackage{todonotes}

\ifpdf
  \hypersetup{
    pdftitle={Adaptive Networks},
pdfauthor={P. S. Chodrow, P. J. Mucha}
  }
\fi
\begin{document}
\maketitle
\textbf{}
\begin{abstract}
	Adaptive voter models (AVMs) are simple mechanistic systems that model the emergence of mesoscopic structure from local networked processes driven by conflict and homophily. 
	AVMs display rich behavior, including a phase transition from a fully-fragmented regime of ``echo-chambers'' to a regime of long-time disagreement governed by low-dimensional quasistable manifolds. 
	Many extant methods for approximating the behavior of AVMs are either restricted in scope, expensive in computation, or inaccurate in predicting important statistics.  
	In this work, we develop a novel, second-order moment closure approximation method for binary-state rewire-to-random and rewire-to-same model variants. 
	We incorporate a small amount of noise via a random mutation term, which renders the system ergodic. 
	Using ergodicity, we then approximate the voting process, which is non-Markovian in the second moments of the system, with a Markovian term near the phase transition. 
	This approximation exploits an asymmetry between different classes of voting events. 
	The resulting scheme enables us to predict the location of the phase transition and the active edge density in the regime of persistent disagreement, across the entire space of parameters and opinion densities. 
	Numerically, our results are nearly exact for the rewire-to-random model, and competitive with extant approaches for the rewire-to-same model. 
	Moreover, our computations display constant scaling in the mean degree, enabling approximations for denser systems than previously possible.  
	We conclude with suggestions for model refinements and extensions.
\end{abstract}

\begin{keywords}
	Networks, nonlinear dynamics, phase transitions, community structure, agent-based simulation
\end{keywords}
\begin{AMS}
	82B26, 
	91D30, 
	91-08, 
	91B14  
\end{AMS}

\section{Introduction} \label{sec:intro}

	A common feature of social networks is trait assortativity, the tendency of similar individuals to interact more intensely or frequently than dissimilar pairs.
	Assortativity can be beneficial, allowing communities of individuals who share common beliefs or experiences to pursue shared goals. 
	On the other hand, assortativity can also restrict flows of information and resources across heterogeneous populations. 
	Recent scrutiny, for example, has fallen on the role of online platforms in promoting political polarization by allowing users to micromanage their contacts and information sources \cite{Anagnostopoulos2014,Bakshy2015}. 
	
	The importance of trait assortativity has inspired various models of self-sorting populations. 
	Among the most influential of these is the classical Schelling model \cite{Schelling1979}, which models the emergence of spatial segregation through a preference of agents to live near a minimum number of similar neighbors. 
	Inspired by this model, the authors of \cite{Henry2011} consider the case of a social network in which agents are assigned an immutable attribute vector that may model demographics or opinions.
	Agents are allowed to destroy their connections to dissimilar partners and create new connections to similar ones, with the aversion to dissimilarity governed by a tunable parameter.
	They then show that the model always generates segregated communities for any nonzero degree of dissimilarity aversion. 
	Because the fixed node attributes are generated exogenously to the system dynamics, this model is most appropriate for studying assortativity based on immutable or slowly-changing attributes, such as demographic variables. 
	The family of voter models \cite{Clifford1973,Holley1975}, which are also defined on networks, provide a contrasting framework.
	In a typical voter model, each node is endowed with an opinion that evolves over time, usually via adoption of the opinion of a uniformly random neighbor.  
	In the original voter model formulations, the network topology is held fixed as opinions evolve. 
	
	In many networks, we naturally expect the opinions of individuals to both influence and be influenced by the connections they form. 	
	Over the past dozen years, a class of \emph{adaptive network} models \cite{vazquez2008generic,Gross2008,Gross2009} has emerged within which to study such interacting influences. 
 	The distinguishing feature of these models is the dynamical coupling between node attributes and network topology.
	Such models have been studied in contexts including epidemic spreading \cite{Gross2006,Marceau2010,Gross2017,Lee2017,Horstmeyer2018} and strategic behavior \cite{Lee2018, pacheco2006coevolution}, but are most commonly deployed as models of opinion dynamics \cite{Holme2006,Durrett2012,Gross2012,Silk2014,Malik2016,Shi2013, pinheiro2016linking}. 
	In this setting, they often appear as \emph{adaptive} (or \emph{coevolutionary}) \emph{voter models} (AVMs).
	AVMs add opinion-based edge-rewiring to the opinion-adoption dynamics of the base voter model.\footnote{Non-voter type updates are also possible in adaptive opinion-dynamics models; see e.g. \cite{Bhawalkar2013a} for a game-theoretic approach.} 
	The tunable coupling of these processes generates polarized networks of opinion-based communities. 
	AVMs thus constitute a class of ``model organisms'' \cite{Silk2014} of endogenous fragmentation, polarization, and segregation in social and information networks. 
	
	Mathematically, AVMs display rich behavior, including metastability and phase transitions. 
	However, the nonlinearity driving this rich behavior renders AVMs difficult to analyze even approximately. 
	Many extant methods are restricted in scope, tractability, or accuracy, and often fail to provide insight into observed behaviors. 
	Our aim in this article is to develop a class of approximation methods that both explain qualitative behaviors in these systems and provide  analytical scope, computational efficiency, and predictive accuracy. 
	
	\subsection{Outline of the Paper}
	    In \Cref{sec:AVMs}, we formulate the class of binary-state AVMs studied here, review their behavior, and survey previous approaches developed for approximating their macroscopic behaviors. 
	    We study a model variant that includes a small amount of random opinion-switching (``mutation''), which renders the model ergodic. 
	    Using ergodicity, we develop in \Cref{sec:analytic} an approximation scheme for the equilibrium macroscopic properties across the entirety of the model's phase space. 
	    Our scheme offers predictions for the point of emergence of persistent disagreement, which corresponds to the ``fragmentation transition'' in non-ergodic model variants.
	    It also offers predictions for the density of disagreement once it emerges, including the arch-shaped quasistable manifolds characteristic of this class of models. 
	    We close in \Cref{sec:discussion} with comparisons to the body of existing models, showing that we achieve favorable scope, accuracy, and computational complexity.  
	    Finally, we discuss promising extensions, both to our approximation methodology and to the model itself. 
	
\section{Adaptive Voter Models} \label{sec:AVMs}
	An adaptive voter model is a first-order, discrete-time Markov process whose states are graphs with opinion-labeled nodes. 
	Each state has the form $\mathcal{G} = (\mathcal{N}, \mathcal{L}, \mathcal{E})$, where $\mathcal{N}$ is a set of nodes and $\mathcal{E}$ a set of edges.
	We let $(u,v) \in \mathcal{E}$ to mean that an edge linking nodes $u$ and $v$ is present in $\mathcal{G}$.
    We denote by $\mathcal{N}(u)$ the neighborhood of $u$, comprising all nodes adjacent to $u$ and $u$ itself. 
    The vector $\mathcal{L}$ maps $\mathcal{N} \rightarrow \mathcal{X}$ where $\mathcal{X}$ is an alphabet of possible states or opinions. 
	We treat the node set $\mathcal{N}$ as fixed, while both $\mathcal{L}$ and $\mathcal{E}$ evolve stochastically in each time-step.
	We here restrict ourselves to the commonly-considered binary-state case, which we denote $\mathcal{X} = \{0,1\}$, though multi-state variants \cite{Holme2006, Shi2013} are also of interest.
	
	The temporal evolution of an AVM is characterized by superimposed voting dynamics on $\mathcal{L}$ and edge-rewiring dynamics on $\mathcal{E}$. 
	To these, our model variant adds a third process in the form of random opinion switching or ``mutation'' in $\mathcal{L}$.
	We specify the discrete-time stochastic dynamics $(\mathcal{E}(t), \mathcal{L}(t)) \mapsto (\mathcal{E}(t+1), \mathcal{L}(t+1))$ as follows:
		\begin{enumerate}
			\item With probability $\lambda \in [0,1]$, \textbf{mutate}: uniformly sample a node $u\in \mathcal{N}$ and set $\mathcal{L}_u(t+1)  \gets \mathrm{uniformChoice}(\mathcal{X}\setminus \{\mathcal{L}_u(t)\})$.
			Note that mutation does not add states to the opinion alphabet $\mathcal{X}$, which is fixed. 
			In the binary-state case, a mutation step deterministically maps $\mathcal{L}_u(t+1)\gets 1 - \mathcal{L}_u(t)$. 
			\item Otherwise (with probability $1-\lambda$),  sample an edge $(u,v) \in \mathcal{E}(t)$ uniformly from the set $\{(u,v):\mathcal{L}_u(t) \neq \mathcal{L}_v(t)\}$ of \emph{active} edges
			(also referred to in some studies as \emph{discordant} edges).
			The orientation of $(u,v)$ is uniformly random. 
			Then, 
			\begin{enumerate}
				\item With probability $\alpha \in [0,1]$, \textbf{rewire}: delete the (undirected) edge $(u,v)$ and add edge $(u,w)$ selected according to one of the following two rules depending on the model variant being used. 
				In the \emph{rewire-to-random} model variant, $w$ is chosen uniformly from $\mathcal{N}\setminus \mathcal{N}(u)$. 
				In the \emph{rewire-to-same} variant, $w$ is chosen uniformly from the set $S_u = \{w \in \mathcal{N}\setminus \mathcal{N}(u) | \mathcal{L}_{w}(t) = \mathcal{L}_u(t)\}$. 
				In the rewire-to-same case, it may in principle occur that $u$ is already connected to all members of the set $S_u$. 
				In this case, we simply pass to the next iteration, starting from Step 1, without modifying $\mathcal{G}$.  
				\item Otherwise (with probability $1-\alpha$) \textbf{vote}:  $\mathcal{L}_u(t+1) \gets \mathcal{L}_v(t)$. 
			\end{enumerate}
		\end{enumerate}
	From a modeling perspective, mutation may represent phenomena such as media influence, noisy communication, or finite agential memory. 
	The mutation mechanism is reminiscent of the ``noisy'' voter model of \cite{GRANOVSKY1995}, and was introduced in an adaptive model variant by \cite{Ji2013}. 

    The rewiring and voting steps both occur after sampling an active edge uniformly at random.
    Other sampling schemes are also possible. 
    The sampling in \cite{Holme2006}, for example, selects a uniformly random node $u$ with nonzero degree.  
    Then, a uniformly random neighbor $v$ of $u$ is chosen. 
    Rewiring occurs with probability $\alpha$ and voting with probability $1-\alpha$ regardless of their respective opinions. 
    In the model introduced by \cite{vazquez2008generic} and further studied by \cite{Toruniewska2017, Ji2013, Kimura2008}, $u$ and $v$ are chosen similarly, but in the event that $\mathcal{L}_{u}(t) = \mathcal{L}_{v}(t)$ nothing happens and the sampling step is repeated. 
    Sampling via active edges as we do here was studied in \cite{Durrett2012} and employed in many recent studies \cite{Demirel2012,Bohme2011, Bohme2012,Silk2014,Basu2015a,Rogers2013}. 
    The authors of \cite{Durrett2012} note that models with different sampling mechanisms nevertheless display similar qualitative --- and often quantitative --- macroscopic behaviors. 
    
    AVMs are usually studied through a standard set of summary statistics.  
	Let $n = \abs{\mathcal{N}}$ be the number of nodes, $m = \abs{\mathcal{E}(t)}$ the number of edges, and $c = {2m}/{n}$ the mean degree.
	Since the dynamics conserve $n$ and $m$, $c$ is time-independent and may be regarded as an additional system parameter. 
	Let $N_i(t) = \abs{\{u \in \mathcal{N} \;|\;\mathcal{L}_u(t) = i \}}$ be the number of nodes holding opinion $i$ at time $t$. 
	Let $\q(t) = (q_0(t), q_1(t)) = n^{-1}\left(N_0(t), N_1(t)\right)$ be the vector of opinion densities.
	For each pair $i$ and $j$ of opinions in $\mathcal{X}$, let $M_{ij}(t) = \abs{\{(u,v) \in \mathcal{E} \;|\; \mathcal{L}_u(t) = i,\; \mathcal{L}_v(t) = j \}}$ be the number of \emph{oriented} edges between nodes of opinion $i$ and nodes of opinion $j$. 
	Note that $M_{ij}(t) = M_{ji}(t)$ and $\sum_{i,j \in \mathcal{X}} M_{ij}(t) = 2m$ at all times $t$, since each (undirected) edge is counted twice in the vector $\mathbf{M}$, once in each of two orientations.
	Let $\X(t) = (X_{00}, X_{01}, X_{10}, X_{11}) = (2m)^{-1}\mathbf{M} = (2m)^{-1}\left(M_{00}(t), M_{01}(t), M_{10}(t), M_{11}(t)\right)$ be the vector of \emph{oriented} edge densities. 
	We define the scalar $R(t) = X_{01}(t) + X_{10}(t) = 2X_{01}(t)$ to be the overall density of  active edges.
	By construction, $R(t)$ is a random variable on the interval $[0,1]$. 
	Let $\x(t) = \E[\X(t)]$ and $\rho(t) = \E[R(t)]$, with expectations taken with respect to the time-dependent measure of the Markov process. 
	Note that the objects $\mathcal{L}(t)$, $\X(t)$, and $R(t)$ are random functions of time $t$, while $\x(t)$ and $\rho(t)$ are deterministic functions of time. 
	
	Most previous studies have considered AVM variants without mutation, corresponding in our setting to $\lambda = 0$. 
	In this setting, any state in which $R = 0$ is an absorbing state of the Markov chain. 
	Such a state consists of one or more connected components within each of which consensus reigns. 
	Letting $C(u)$ denote the connected component of node $u$ in the absorbing state in this regime, it holds that $C(u) = C(v)$ implies $\mathcal{L}_u = \mathcal{L}_v$ for any nodes $u$ and $v$.  
	As discussed in both \cite{Holme2006} and \cite{Durrett2012}, there is a phase transition in the (random) final value $\q^*$ of the opinion densities in the absorbing state.
	In both model variants, there is a critical value $\alpha^*$, depending on $\q(0)$, such that, if
	$\alpha \geq \alpha^*(\q(0))$, $\abs{\q^* - \q(0)}_1 = O\left(\frac{\log n}{n}\right)$ with high probability as $n$ grows large. 
	Thus, in the large $n$ limit, the opinion densities are not appreciably altered by the dynamics. 
	We refer to this as the ``subcritical'' parameter regime. 
	On the other hand, if $\alpha < \alpha^*(\q(0))$, $\q^*$ is governed by a bimodal distribution parameterized by $\alpha$, $c$, and the model rewiring variant. 
	In both models, the phase transition marks the point at which the voting dynamics outstrip the rewiring dynamics, in the sense that the rewiring dynamics are no longer fast enough to resolve most disagreements, and therefore also corresponds to a transition in the time to reach the final state \cite{Holme2006,Rogers2013}. 
	We refer to this regime as ``supercritical.''
	Note that in the $\lambda = 0$ case, the non-ergodicity of the model implies that all these results are to an extent dependent on the initial state $\mathcal{G}_0$ of the AVM. 
	While this dependence is generally weak in the standard AVM we consider here, in related model variants \cite{kureh2019fitting} the initial conditions may often dominate even the rewiring mechanism in determining the final system state. 
	
	In \cite{Durrett2012}, the authors show via simulation and analytical approximations that the same phase transition marks the emergence of a \emph{quasistable manifold} along which the system dynamics evolve. 
	This manifold is well-approximated by a concave parabola in the $(q_1,\rho)$-plane, reflected by its colloquial name, ``the arch.'' 
	Similar arches were observed for an AVM variant in \cite{vazquez2008generic} and for a non-adaptive voter model in \cite{vazquez2008analytical}. 
	When $\alpha > \alpha^*(\q(0))$, $\rho$ converges rapidly to $0$ while $\q$ remains nearly constant.
	When $\alpha < \alpha^*$, on the other hand, the trajectory converges to a point on the arch, and then slowly diffuses along it until reaching an absorbing state at one of the two bases (zeros) of the arch.  
	In the rewire-to-random arch, $\alpha^*$ depends on $q_1$ in that the arch is supported on a proper sub-interval of $[0,1]$. 
	In contrast, the rewire-to-same transition is independent of $q_1$, with the associated arch supported on the entirety of $[0,1]$.
	The bases of the arch correspond to the modes in the long-run distribution of $\q^*$. 
	
	\begin{figure}
	 	\centering
	 	\includegraphics[width=\textwidth]{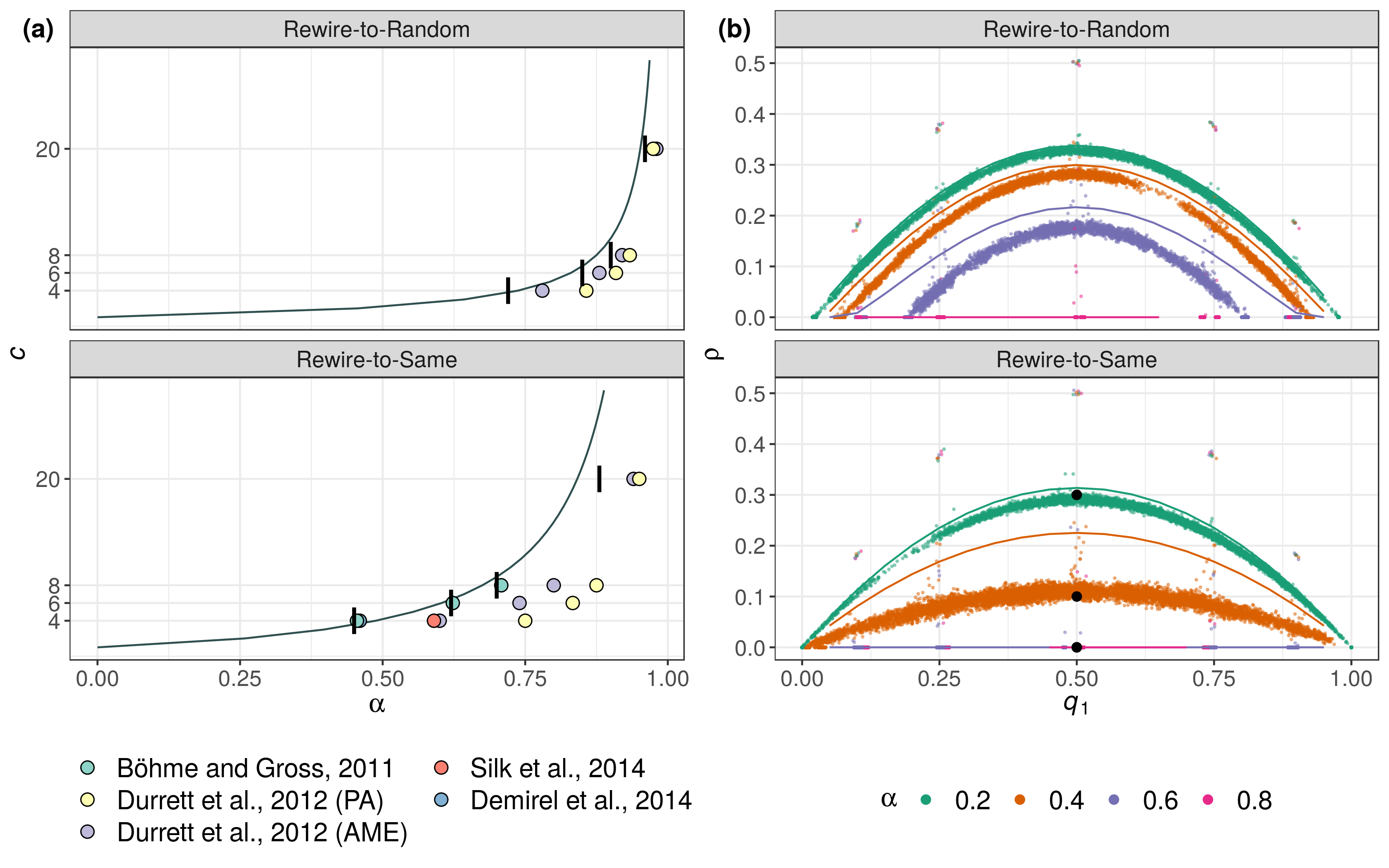}
	 	\caption{
		 	(a): Estimates of the phase transition $\alpha^*$ in the expected density $\rho$ of active edges when $\q(0) = \left(\frac{1}{2},\frac{1}{2}\right)$, for varying mean degree $c$. 
		 	Vertical black lines give the empirical location of the phase transition, determined numerically as the smallest value of $\alpha$ for which $\rho > 0.01$. 
		 	Colored points give estimates of the phase transition from extant methods. 
		 	The $c = 4$ estimates for the rewire-to-same variant of B\"ohme and Gross \cite{Bohme2011} and Demirel et al.~\cite{Demirel2012} overlap. 
		 	The solid line gives the estimate of our proposed method, obtained by solving \Cref{eq:transition_approx}.
		 	(b): Quasi-stable arches in the $(q_1,\rho)$-plane for varying $\alpha$.
		 	Higher arches correspond to lower values of $\alpha$.
		 	Points are sampled from simulations at intervals of $5,000$ time-steps.
		 	Black dots on the rewire-to-same panel give the active-motif estimate of \cite{Demirel2012} for the symmetric top of the arch.
		 	The solid lines give the approximate master equation estimates of \cite{Durrett2012}. 
		 	Simulations in this and subsequent figures were performed with $N = 10^{4}$ nodes and mutation rate $\lambda = 2^{-10} \approx 10^{-3}$. 
		 	All simulations were initialized with an Erd\H{o}s-R\'enyi $G(n,p)$ graph on  $n = 10^{4}$ nodes with specified mean degree $c$.  
		 	Each node independently chooses its initial opinion $0$ or $1$ with equal probability. 
		 	We performed $10^7$ total simulation steps, sampling the process after a burn-in period of $2\times 10^6$ steps. 
		 	This process was repeated ten times for each combination of parameters $\alpha$ and $c$. 
	 	} 
	 	\label{fig:prev_results}
	\end{figure} 
	
	While multiple studies have achieved insight via numerical study of simulation traces \cite{Yi2013, Shi2013, Ji2013},  analytical insight into the phenomenology of AVMs remains relatively limited.
	The central analytical project is to estimate the behavior of the expected edge density $\rho$ as a function of the parameters $\lambda$, $\alpha$, and $c$, as well as the opinion density $\q$.\footnote{Recent papers have studied other features of interest, such as approximate conservation laws \cite{Toruniewska2017} and network topology near the phase transition $\alpha^*$ \cite{Horstmeyer2018}; however, we will not pursue these themes further here.}
	The most modest task is to estimate the phase transition $\alpha^*$ in the case of symmetric opinion initialization $\q(0) = \left(\frac{1}{2}, \frac{1}{2}\right)$. 
	\Cref{fig:prev_results}(a) summarizes a selection of extant methods developed over the past decade to approximate the location of the phase transition in these model variants and compares them to the observed emergence of the top of the arch in model simulations. 
	The pair approximation (PA) \cite{Kimura2008, Durrett2012} is an all-purpose method for binary-state models that usually produces qualitatively correct but quantitatively poor results.
	Indeed, \Cref{fig:prev_results} shows that the pair approximation overestimates the location of the phase transition, performing especially poorly in the rewire-to-same model variant. 
	More specialized methods are required to obtain quantitatively reasonable estimates. 
	The method of \cite{Bohme2011} uses compartmental equations to accurately estimate the rewire-to-same phase transition with symmetric opinion densities, finding close agreement with observation in this restricted task. 
	In \cite{Basu2015a} the authors apply stopping-time arguments to give a rigorous proof of the existence of a phase transition in both model variants. 
	However, their results apply only in the context of dense limiting graphs and do not explicitly predict the value of $\alpha^*$. 
	
	Other schemes provide estimates not only of the transition but also of the quasistable supercritical active link density $\rho$ when $\q(0) = \left(\frac{1}{2}, \frac{1}{2}\right)$.
	The authors of \cite{Demirel2012} propose a compartmental approach based on \emph{active motifs} to estimate the phase transition and arch in the symmetric opinion rewire-to-same model variant.
	An active motif consists of a node and a number of active links attached to it; a system of ordinary differential equations may be obtained by approximately tracking the evolution of active motif densities in continuous time. 
	The resulting estimate of the phase transition (\Cref{fig:prev_results}(a)) and of the top of the arch (\Cref{fig:prev_results}(b)) are both highly accurate, but require an active-link localization assumption specific to the rewire-to-same variant. 
	The authors of \cite{Silk2014} follow a related approach for the rewire-to-same variant based on more general \emph{active neighborhoods}.
	Active neighborhoods count the numbers of both active and inactive links attached to a given node. 
	They obtain an analytic approximation by transforming the resulting system into a single partial differential equation governing the generating functions of the neighborhood densities. 
	The resulting estimate of the phase transition (\cref{fig:prev_results}(a)) and the active link density (not shown) are, however, uniformly dominated in accuracy by the explicit active-motif approach. 
	
	To our knowledge, the only methods for  approximating the complete arch are the pair approximation and the approximate master equations (AMEs, \cite{Gleeson2011}) used in \cite{Durrett2012}. 
	Approximate master equations are similar to active-neighborhood techniques, but are formulated explicitly for the case of general opinion densities $\q$. 
	For small mean degree, approximate master equations can provide relatively accurate predictions of the rewire-to-random phase transition (\Cref{fig:prev_results}(a)) and qualitatively reasonable estimates of the arches (\Cref{fig:prev_results}(b)), though the shapes of the arches may be somewhat distorted. 
	Their estimates for $\alpha^*$ and $\rho$ in the rewire-to-same variant are substantially worse, although the qualitative shape of the arches appears correct. 
	AMEs are constrained by their computational cost: to obtain a solution requires the numerical solution of $\Theta(k_\mathrm{max}^2)$ coupled differential equations, where $k_\mathrm{max}$ is the largest node-degree expected to emerge in the course of a simulation, and therefore depends at least linearly on the mean degree $c$. 
	The scheme thus rapidly becomes impractical for high enough mean degree or for initially skewed degree distributions.

	\subsection{AVMs with Mutation}
	
	The approximation scheme we will develop in \Cref{sec:analytic} depends on the presence of mutation in the model, i.e.\ $\lambda > 0$. 
	The introduction of mutation has an important technical consequence: the process is ergodic, up to symmetry. 
    \begin{dfn}
        A \emph{labeled graph isomorphism} of a state $\mathcal{G} = (\mathcal{N}, \mathcal{L}, \mathcal{E})$ is a permutation $\tau:\mathcal{N}\rightarrow \mathcal{N}$ such that $(u,v) \in \mathcal{E}$ iff $(\tau(u),\tau(v)) \in \mathcal{E}$ and $\mathcal{L}_u = \mathcal{L}_{\tau(u)}$ for all $u \in \mathcal{N}$.
        We write  $\overline{\mathcal{G}}$ for the equivalence class of $\mathcal{G}$ under labeled graph isomorphisms. 
    \end{dfn}
    
    \begin{theorem}
        When $\lambda > 0$, if $\binom{n-4}{2} \geq m-1$, the process $\overline{\mathcal{G}}(t)$ is ergodic. 
    \end{theorem}
    \begin{proof}
        We will first show aperiodicity by constructing cycles of lengths $2$ and $3$ in the state space. 
        To construct a cycle of length $2$, simply choose a node and perform two sequential mutation steps. 
        The construction of a cycle of length $3$ is slightly more involved. 
        Pick an edge $e \in \mathcal{E}$. 
        Label one end $u$, and the other end $v_1$. 
        Pick two more nodes $v_2$ and $v_3$. 
        Using mutation and rewiring steps, remove all edges connected to $u$, $v_1$, $v_2$, and $v_3$ except for $e$. 
        This can always be done by hypothesis, since the remaining $m-1$ edges may be placed among the $\binom{n-4}{2}$ pairs of remaining nodes via mutation and rewiring steps.  
        Using mutation steps, set $\mathcal{L}_u = \mathcal{L}_{v_2} = 0$ and $\mathcal{L}_{v_1} = \mathcal{L}_{v_3} = 1$. 
        Call this initial state $\mathcal{G}$. 
        Then, consider the following sequence: 
        \begin{enumerate}
            \item Rewire $(u,v_1) \mapsto (u,v_2)$.
            \item Mutate $\mathcal{L}_{v_2} \gets 1$.
            \item Mutate $\mathcal{L}_{v_1} \gets 0$.
        \end{enumerate}
        Call the end state $\mathcal{G}'$.
        Each of these steps is supported in both rewire-to-same and rewire-to-random model variants. 
        Furthermore, the permutation $\tau$ that interchanges $v_1$ and $v_2$ is a labeled isomorphism from $\mathcal{G}$ to $\mathcal{G}'$. 
        We have therefore constructed a supported cycle of length 3 in the state space of the process $\overline{\mathcal{G}}(t)$, completing the proof of aperiodicity.  
        
        To show irreducibility, let  $\mathcal{G}_1 = (\mathcal{N}, \mathcal{L}_1, \mathcal{E}_1)$ and $\mathcal{G}_2 = (\mathcal{N}, \mathcal{L}_2, \mathcal{E}_2)$ be elements of the state space of a single AVM.
        Since $\abs{\mathcal{E}_1} = \abs{\mathcal{E}_2} = m$, we have $\abs{\mathcal{E}_1 \setminus \mathcal{E}_2} = \abs{\mathcal{E}_2 \setminus \mathcal{E}_1}$. 
        These sets may therefore be placed in bijective correspondence. 
        For each edge $e = (u,v) \in \mathcal{E}_1\setminus \mathcal{E}_2$, we arbitrarily identify $e' = (u',v') \in \mathcal{E}_2 \setminus \mathcal{E}_1$. 
        Perform the sequence of rewirings $(u,v)\mapsto (u, v') \mapsto (u', v')$  possibly with mutation steps in order to activate the edges as needed. 
        Doing so reduces the set $\mathcal{E}_1\setminus \mathcal{E}_2$ by one edge. 
        Repeat this process inductively until $\mathcal{E}_1\setminus \mathcal{E}_2 = \emptyset$; that is, until $\mathcal{E}_1 = \mathcal{E}_2$. 
        Finally, perform mutation steps on all nodes $u$ on which $\mathcal{L}_1$ and $\mathcal{L}_2$ disagree. 
        The result is a path of nonzero probability through the state space of $\mathcal{G}(t)$ and therefore of $\overline{\mathcal{G}}(t)$, as was to be shown.
    \end{proof}
    
    Since the process $\overline{\mathcal{G}}(t)$ is ergodic, it possesses an equilibrium measure $\eta$ supported on the entirety of its state space. 
    In the remainder of this paper, we will abuse notation by identifying $\mathcal{G}$ with $\bar{\mathcal{G}}$ and referring to $\eta$ as the equilibrium distribution of $\mathcal{G}(t)$.
    An important consequence of ergodicity is that all properties of $\eta$ derived for the mutating AVM are independent of the state in which the network is initialized.
    This situation contrasts with the extant literature discussed previously (with the exception of \cite{Ji2013}), in which the model behavior can, in principle, depend on initialization. 
    Ergodicity also implies that states with $R = 0$ are no longer absorbing. 
	Instead, a typical sample from $\eta$ displays bifurcated structure closely aligned with the opinion groups, with dense connections between common opinions and sparser connections between differing opinions.
	This behavior of the mutating AVM thus makes it a more flexible model of social processes in which  long-standing disagreement may influence connections. 
	We focus here on the limit of small $\lambda$, which allows us to derive approximations for the non-mutating AVMs. 
	In particular, the equilibrium measure $\eta$ concentrates around the $\lambda=0$ arch, allowing us to describe the arch as the expected active link density $\rho^* = \E_\eta[R]$ under the equilibrium measure $\eta$. 
	
	\section{Model Analysis} \label{sec:analytic}
        Our strategy is to study perturbations from the fully-fragmented state $R = 0$. 
        These perturbations are induced by mutation, without which the fully-fragmented state is absorbing. 
        While many existing techniques amount to continuous-time mass-action laws for system moments, our methods are fundamentally discrete and local in that we study changes in the edge density vector $\X$ stemming from a single mutation event. 
       	Carefully-chosen approximations in this regime can be expected to be accurate near the critical point. 

        Assume that $\lambda$ is small but positive. 
        Suppose that at time $t$, $R = 0$. 
        In this state, $\mathcal{G}(t)$ consists of one or more connected components within which the opinion function $\mathcal{L}$ is constant. 
        Suppose now that, at time $t+1$,  node $u$ on component $C(u)$ changes its opinion from $0$ to $1$ through mutation. 
        Because opinions on $C(u)$ are otherwise uniform, all active links present in component $C(u)$ are contained in the neighborhood of $u$ itself. 
        In particular, any additional active links that may be generated over short timescales near $u$, as measured by the geodesic distance on $\mathcal{G}$. 
        
        Let $T$ be the hitting time of the event $R = 0$; i.e., the amount of time required to return to the fully-fragmented state. 
        We can distinguish two regimes, depending on the scaling of $\E[T]$ with $n$.  
        \begin{enumerate}
            \item \textbf{Subcritical}: 
            We have $\E[T] = O(1)$. 
            Intuitively, this occurs when $u$'s dissenting opinion is either snuffed out by voting events or ``quarantined'' by rewiring events in a small number of time steps. 
            This case always occurs when $\alpha = 1$, since $T$ is then simply the time until each active link has been rendered inactive via rewiring. 
            The expected number of active edges scales as $n\rho^* = O(1)$, since there are only $O(1)$ time steps in which additional active edges may be generated.  
            We therefore have $\rho^* \rightarrow 0$ as $n$ grows large. 
            \item \textbf{Supercritical}: 
            We have $\E[T] = O(n^2)$, corresponding to the consensus-time of the non-adaptive voter model \cite{Holley1975}; 
            as such, this case always occurs for $\alpha = 0$. 
            Mechanistically, $u$'s dissenting opinion triggers a cascade of active edge generation through voting and rewiring events with nonzero probability.  
            In this case, the number $R$ of active edges scales with $n$ (see, e.g. \cite{vazquez2008analytical}), and the equilibrium active edge density $\rho^*$ is nonzero as $n$ grows large. 
        \end{enumerate}
        These two regimes are separated by critical values in the parameters $\alpha$, $\lambda$, and $c$. 
        Indeed, the transition in $\alpha$ is precisely that described previously for the $\lambda = 0$ case. 
        The situation is thus reminiscent of the standard Galton-Watson branching process \cite{athreya2004branching}, in which the criticality of the aggregate process can be characterized locally by the reproductive potential of a single node. 
        
        To develop quantitative approximations, we therefore study the local dynamics around node $u$. 
        At the moment that node $u$ changes its opinion from $0$ to $1$, all nodes on $C(u)$ other than $u$ itself have opinion $0$. 
        Even if further mutation events take place, it will still be true that local neighborhoods of nodes in $C(u)$ are statistically dominated by opinion $0$ nodes. 
        That is, we can distinguish a local minority --- initially comprising node $u$ alone --- of opinion $1$ nodes. 
        In the subcritical regime, every connected component possesses a clearly defined local minority and local majority. 
        In the supercritical regime, these distinctions degrade as the number of active links increases. 
    
        We will use this physical intuition to formulate a closed-form approximation in the neighborhood of the critical point. 
        Then, the dynamics in the expected edge counts may be written 
        \begin{align}
    		\mathbf{m}(t+1) - \mathbf{m}(t) =\E\left[\lambda \mathbf{W}(\mathcal{G}(t)) + (1-\lambda)\alpha \mathbf{R}(\mathcal{G}(t)) + (1-\lambda)(1-\alpha) \mathbf{V}(\mathcal{G}(t))\right]\;, \label{eq:dynamics}
    	\end{align}
    	where $\mathbf{W}$, $\mathbf{R}$, and $\mathbf{V}$ are (random) functions of the graph state $\mathcal{G}(t)$ giving the increments in $\mathbf{m}$ due to mutation, rewiring, and voting, respectively. 
        Importantly, $\mathbf{W}$ and $\mathbf{R}$ depend only on $\q$ and $\x$, the expected first and second moments of $\mathcal{L}$. 
    	Starting with the former, the entries of the expected mutation term may be written 
    	\begin{align}
    	    \E[\mathbf{W}(\mathcal{G})] = \mathbf{w}(\x) = c 
    	    \left[
    	        \begin{matrix}
    	        x_{10} - x_{00} \\ 
    	        x_{00} - x_{10} + x_{11} - x_{01}\\
    	        x_{00} - x_{10} + x_{11} - x_{01}\\
    	        x_{01} - x_{11}
    	        \end{matrix}
    	    \right]\,.
    	\end{align}
    	We illustrate by deriving the expression for $\mathbf{w}_{00}(\mathbf{x})$. 
    	Edges between nodes of opinion $0$ are created when an opinion-1 node on an active edge mutates. 
    	A uniformly random opinion-1 node has in expectation $cx_{10}$ active edges available to transform into $0$-$0$ edges upon mutation. 
    	Similarly, $0$-$0$ edges are destroyed when one of the incident nodes mutates. 
    	A uniformly random opinion-0 node has in expectation $cx_{00}$ inactive edges that are destroyed upon mutation. 
    	The expressions for the other entries of $\mathbf{w}$ are derived by parallel arguments. 
    	The  rewiring terms $\mathbf{r}$ are written as follows:
    	\begin{align*}
    		\E[\mathbf{R}(\mathcal{G})] = \mathbf{r}(\q) = 
    		\begin{cases}
    			\left[q_0, -\frac{1}{2}, -\frac{1}{2}, q_1\right]^T & \text{rewire-to-random}\\
    			\left[1, -1, -1, 1\right]^T & \text{rewire-to-same}.
    		\end{cases}
    	\end{align*}
    	Notably, the rewiring function depends on the opinion densities $\q$ only in the rewire-to-random case, because the rewire-to-same variant always removes exactly one active edge, replacing it with an inactive one in a rewiring step. 
    	To derive the expression for the rewire-to-random case, we can condition on the opinion of the node that ``keeps'' the active edge $e$. 
    	If the $0$-opinion node keeps $e$, then, with probability $q_0$, $e$ joins to another opinion $0$ node, destroying the active edge and creating a $0$-$0$ edge. 
    	A similar argument accounts for the $q_1$ term. 
    	Summing up the ways for an active edge to be removed, we have $r_{01}(\q) = -\frac{1}{2}\left(q_0 + q_1\right) = -\frac{1}{2}$, as was to be shown. 
    	
    	The computations above show that the mutation and rewiring dynamics in $\X$ are Markovian: for any fixed $\q$, when $\alpha = 1$, $\mathbf{X}$ is a Markov process. 
    	Because of this, computing $\x$ in the $\alpha = 1$ case for fixed $\q$ reduces to solving a four-dimensional linear system subject to nonnegativity and normalization constraints. 
    	Unfortunately, the voting term $\mathbf{v}(\mathcal{G}(t)) = \E[\mathbf{V}(\mathcal{G})]$ cannot be similarly parsed in terms of $\q$ and $\X$, because the voting dynamics depend on higher graph moments and are therefore non-Markovian in these variables. 
    	We may therefore view the short-timescale dynamics of $\X$ for fixed $\q$ as a mixture of Markovian opinion-switching and rewiring processes with a non-Markovian voting process. 
    	Our strategy is to approximate the expectation of the non-Markovian voting term with a Markovian approximation near the phase transition, using the asymmetry between local minorities and majorities. 
    	This approximation supposes that $\E[\mathbf{V}(\mathcal{G}(t))] \approx \hat{\mathbf{v}}(\q, \x)$ near the critical regime, with the function $\hat{\mathbf{v}}$ of $\q$ and $\x$ to be determined. 
    	Note that doing so and setting the lefthand side of \cref{eq:dynamics} equal to zero removes all dependence on the time-step $t$ except dependence through $\q$ and $\x$.
    	We therefore suppress the argument $t$ throughout the remainder of this section. 
    	All expectations are taken with respect to the stationary distribution $\eta$.

    	To construct $\hat{\mathbf{v}}$, we study the local neighborhood of a node $u$ that has just changed its opinion from $\bar{\imath} \in \{0,1\}$ to $i \in \{1,0\}$.
    	We denote expectations conditioned on this event using the shorthand $\E[\cdot|i]$.
    	Immediately after this event, $u$ possesses an initial random number $J_0$ of inactive and $K_0$ of active edges. 
    	The distributions of $J_0$ and $K_0$ depend on $\q$, $\x$, $c$, and their moments, as well as the conditions under which node $u$ changed its opinion. 
    	If $u$ changes its opinion due to a mutation on an otherwise constant-opinion component, then $J_0 = 0$. 
    	On the other hand, if $u$ changes its opinion through a voting event, then $J_0 \geq 1$, since there must have been a node to pass on the opinion to $u$. 
    	To compute $\hat{\mathbf{v}}$, we track each of the $K_0$ active edges until each of them has been rendered inactive, counting voting events along the way. 
    	
    	Under timescale-separation and mean-field assumptions, these calculations may be carried out in closed form. 
    	The assumption of timescale-separation supposes that $\mathcal{G}$ changes slowly relative to the neighborhood of node $u$, so that only update steps that sample the initial $K_0$ edges require accounting. 
    	The mean-field assumption supposes that nodes in the local majority have degree distributions governed by the global network average $\x$, reflecting the fact that, by definition, most nodes are members of their respective local majorities. 
    	These assumptions are approximately correct when the active edge density $R$ and mutation rate $\lambda$ are both small, and will tend to degrade when either quantity is increased. 
    	
    	Define the vector $\mathbf{c}$ with components $c_{ij} = cx_{ij}/q_i$.
    	Each entry $c_{ij}$ gives the average number of neighbors of type $j$ of a node of type $i$. 
    	Note that, though $x_{ij} = x_{ji}$, it is not generally the case that $c_{ij} = c_{ji}$ unless $\q = \left(\frac{1}{2}, \frac{1}{2}\right)$. 
    	The random variable $K_0$ is the number of opinion $\bar{\imath}$ neighbors incident to $u$ immediately prior to $u$ changing opinion; under the mean-field assumption, we therefore have $\E[K_0|i] = c_{\bar{\imath}\bar{\imath}}$. 
    	Meanwhile, $\E[J_0|i] = 1 + c_{\bar{\imath} i}$ if $u$ changed its opinion due to voting and $\E[J_0|i]=0$ if $u$ changed its opinion due to mutation. 
    	Since we assume $\lambda$ to be small and mutations to therefore be slow, we focus on the former case. 
    	
    	We need to track multiple types of voting events, and we define random variables for each. 
    	\begin{enumerate}
    	    \item Neighbors of $u$ may vote. By the assumption of timescale separation, each such vote occurs along one of the $K_0$ initial active edges. Let $E$ denote the (random) number of such votes. 
    	    \item Nodes not attached to $u$ may vote. In the rewire-to-random model, such events may occur after an active edge attached to $u$ is rewired away from $u$ but remains active, allowing for
    	    a later time at which one of the two nodes on this edge votes to render the edge inactive.
    	    In the rewire-to-same model variant, this type of voting event does not occur. 
    	    Let $F$ denote the (random) number of such voting events.
    	    \item Node $u$ itself may vote prior to all of its $K_0$ active edges becoming inactive or removed from $u$. Let $G$ denote the indicator random variable for this event. 
    	\end{enumerate}
    	We next write down vectors tracking the impact of each of the above  voting event types on $\mathbf{m}$, the vector of expected global edge counts. 
    	We first compute the vector $\mathbf{e}_i(\C)$ whose entries give the expected increment in $\mathbf{m}$ due to a Type 1 event.
    	Since votes occur along active edges, a Type 1 event consists of a neighboring node $v$ changing opinion from $\mathcal{L}_v = \bar{\imath}$ to $\mathcal{L}_v = i$. 
    	In this event, edge $(u,v)$ is rendered inactive. 
    	At node $v$, $c_{\bar{\imath}\bar{\imath}}$ edges are activated in expectation, and $c_{\bar{\imath}i}$ edges are rendered inactive as $i$-$i$ edges. 
    	Both of these expressions are implied by the mean-field approximation. 
    	We therefore have 
    	\begin{align} 
    		\mathbf{e}_i(\C) &= \frac{1}{2}\left(-2\E[K_0|i] , \E[K_0|i] - \E[J_0|i], \E[K_0|i] - \E[J_0|i], 2\E[J_0|i])\right) \nonumber\\
    		&= \frac{1}{2}\left(-2c_{\bar{\imath}\bar{\imath}} , c_{\bar{\imath}\bar{\imath}} - c_{\bar{\imath}i} - 1, c_{\bar{\imath}\bar{\imath}} - c_{\bar{\imath}i} - 1, 2(1+c_{\bar{\imath}i})\right)\;. \label{eq:neighbor_vote}
    	\end{align}
    	The expected increment vector for Type 2 events may again be computed via the mean-field approximation.
    	Since the edges involved in Type 2 events are not connected to $u$, the increment is independent of $\mathcal{L}_u$.  
    	We therefore have 
    	\begin{align} \label{eq:wild_vote}
    		\mathbf{f}(\C) = \frac{\mathbf{e}_0(\C) + \mathbf{e}_1(\C)}{2}\;.
    	\end{align}
    	The analysis for Type 3 events is more subtle. 
    	For $i = 1$, this term has components
    	\begin{align}
    		\mathbf{g}_1(\q, \C) = \frac{1}{2}\left(2\E[GK|1], \E[G(J-K)|1], \E[G(J-K)|1], -2\E[GJ|1]\right)\;, \label{eq:g_components}
    	\end{align}
    	where $J$ (respectively, $K$) are the number of inactive (active) edges attached to $u$ at the time of voting, and $G$ is the event that $u$ votes prior to deactivation. 

    	To complete the approximation scheme, it is necessary to first compute the expectations appearing in \Cref{eq:g_components} and then compute the expected number of events of each type.
    	We begin with $K$, the active edge count at the time that $u$ votes. 
    	Conditioned on a fixed initial number $K_0$ of active edges and $u$'s opinion $i$, $K$ is distributed as a truncated geometric random variable: 
         \begin{align*}
            \prob(K = k|i, K_0) = 
            \begin{cases}
              (1-\beta_i)\beta_i^{K_0 - k} &\quad 1\leq k \leq K_0\\ 
              \beta_i^{K_0} &\quad k = 0,
            \end{cases}
        \end{align*}
        where $\beta_i$ is the probability that an event is not a vote by $u$, given that it removes a discordant edge from $u$ and that $u$ has opinion $i$. 
        This probability is given explicitly by 
        \begin{align}
            \beta_i =
            \begin{cases}
              \frac{1+\alpha q_i}{2-\alpha(1-q_i)} &\quad \text{rewire to random}\\ 
              \frac{1+\alpha}{2} &\quad \text{rewire to same.} \label{eq:beta}
            \end{cases}
        \end{align}
        To derive the rewire-to-random expression, we  enumerate the events that remove an active edge $(u,v)$ from $u$, given that $(u,v)$ is sampled for update. 
        A vote by either node $u$ or node $v$ deactivates the edge, and occurs with probability $1-\alpha$. 
        A rewiring event in which $v$ maintains the edge removes the edge from $u$ and occurs with probability $\alpha/2$. 
        A rewiring event in which $u$ maintains the edge occurs with probability $\alpha/2$, and deactivates the edge with probability $q_i$ in the rewire-to-random case. 
        The total rate of active edge removal from $u$ is therefore $2-\alpha(1-q_i)$. 
        The rate of active edge removal, excluding Type 3 voting events, is $2-\alpha(1-q_i) - (1-\alpha) = 1+\alpha q_i$. 
        A similar derivation yields the expression for the rewire-to-same variant. 
        
        The probability of $u$ voting prior to deactivation, conditioned on $K_0$, is 
        \begin{align*}
            \E[G|K_0, i] = \prob(K \geq 1) = 1-\beta_i^{K_0}\;. 
        \end{align*}
        Averaging over $K_0$ yields 
        \begin{align*}
            \E[G|i] = \sum_{k_0}\prob(K_0 = k_0)(1-\beta_i^{k_0}) = 1 - \phi_{K_0}(\beta_i)\;,
        \end{align*}
        where $\phi_{K_0}(z) = \sum_{k = 1}^\infty \prob(K_0 = k) z^{k}$ is the probability generating function of $K_0$.
        Some previous work (e.g. \cite{vazquez2008generic}) explicitly models quantities such as  $K_0$ as binomial or Poisson random variables.
        In our experiments, the crude approximation $\E[G|i] \approx 1 - \beta_{i}^{\E[K_0|i]} = 1 - \beta_{i}^{c_{\bar{\imath}\bar{\imath}}}$ yields similar results with much faster computations, and is therefore used in the results presented below. 
        
        The expected number of active edges at the time that $u$ votes is 
        \begin{align*}
        \E[GK|i] &= \E_{K_0}\E[GK|i, K_0] \\ 
                 &= \E_{K_0}\left[K_0 - \frac{\beta_i(1-\beta_i^{K_0})}{1-\beta_i}\bigg|i\right] \\ 
                 &= \E[K_0|i] - \frac{\beta_i}{1-\beta_i}\E[G|i]\;.
        \end{align*}
        This accounts for the decay in the local active edge density around $u$. 
        It remains to compute $\E[E|i]$, $\E[F|i]$, and $\E[GJ|i]$. 
        To do so, it is useful to introduce the coefficients
        \begin{align}
            \varepsilon_i =
            \begin{cases}
              \frac{1-\alpha(1-q_i)}{1+\alpha q_i}\,,  \\ 
              \frac{1}{1+\alpha}\,,
            \end{cases} \quad 
            \sigma_i = 
            \begin{cases} q_1\frac{2(1-\alpha)}{2-\alpha}\,, &\quad \text{rewire to random}\,, \\ 
              0\,,   &\quad \text{rewire to same.}
            \end{cases} \label{eq:coefs}
        \end{align}
        The coefficient $\varepsilon_i$ gives the probability that an event that removes an active edge from $u$, other than a vote by $u$, produces an inactive edge either through rewiring or through a vote by a neighbor of $u$. 
        The coefficient $\sigma_i$ gives the probability that an active edge which is rewired but not immediately deactivated is ultimately deactivated via a voting event. 
        The derivations of these coefficients are similar to that of $\beta_i$ above. 
        
        Node $u$ begins with an initial number $J_0$ of inactive edges, and gains more via  rewiring and voting. 
        At the time that $u$ votes, in expectation $\E[K_0|i] - \E[GK|i]$ active links have been removed; each has a probability $\varepsilon_i$ of being deactivated while remaining attached to $u$. 
        The expected number of inactive edges at the time that $u$ votes is therefore 
         \begin{align*}
            \E[GJ|i] &= \E[J_0|i] + \varepsilon_i \left(\E[K_0|i] - \E[GK|i]\right)\;.
        \end{align*}
        To compute $\E[E|i]$, the expected number of Type 1 events, we note that a voting event along edge $(u,v)$ has equal probability to change $\mathcal{L}_u$ as $\mathcal{L}_v$. 
        The expected number of Type 1 events is therefore equal to the expected number of Type 3 events, and we have $\E[E|i] = \E[G|i] = 1 - \phi_{K_0}(\beta_i)$. 
        Finally, we compute the expected number of Type 2 events. 
        By definition, for a Type 2 event to occur, the edge must no longer be attached to $u$.
        The expected number of such edges is $\E[K_0 + J_0 - G(K + J)|i]$. 
        The probability that such an edge was removed by $u$ by a rewiring event that did not deactivate the edge is $\sigma_i$. 
        We obtain
        \begin{align*}
            \E[F|i] = \sigma_i\E[K_0 + J_0 - G(K + J)|i]\;.
        \end{align*}
        
        \begin{figure}
    	 	\centering
    	 	\includegraphics[width=\textwidth]{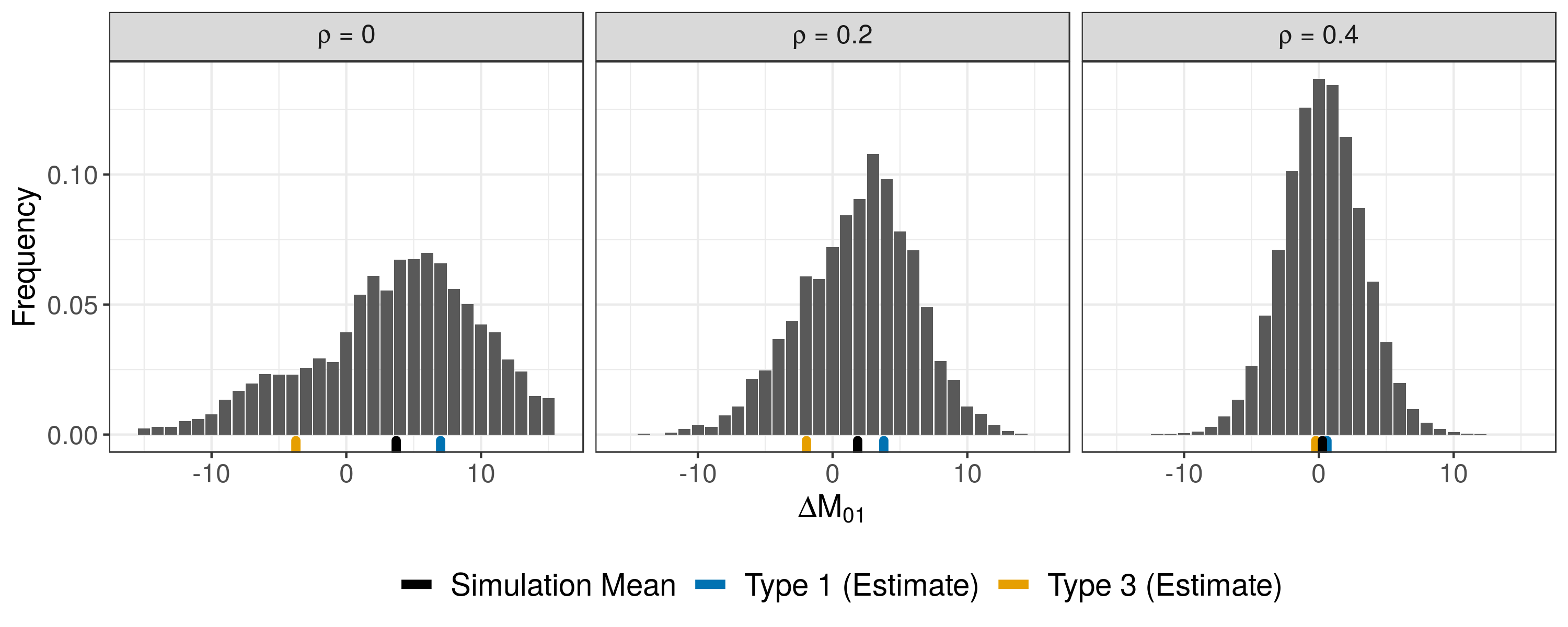}
    	 	\caption{Illustration of the asymmetry between Type 1 and Type 3 events. 
    	 	Histograms give the impact of a voting event on $M_{01}$, the number of active edges. 
    	 	Each panel corresponds to a different value of the expected active edge density $\rho$. 
    	 	The expected impact of Type 1 (blue) and Type 3 (orange) events are shown in the horizontal margin, as well as the simulation mean (black). 
    	 	Simulations performed on a rewire-to-random AVM of $n = 10^4$ nodes and $c = 8$, with varying rewiring rate $\alpha$ under the conditions described in \Cref{fig:prev_results}.  
    	 	Events are tallied only for $0.45 \leq q_1 \leq 0.55$.} 
    	 	\label{fig:symmetry_breaking}
    	\end{figure} 
        
        An important prediction of this formalism is that Type 1 and Type 3 events, though they occur at the same rate, have different impacts on the active edge density. 
        Since $\E[K|i] < \E[K_0|i]$ and $\E[J|i] > \E[J_0|i]$, we have 
        \begin{align} \label{eq:differential_impacts}
            \mathbf{e}_i(\C)_{\bar{\imath}i} = \frac{\E[K_0|i] - \E[J_0|i]}{2} > \frac{\E[K|i] - \E[J|i]}{2} =   
            -{g}_i(\q,\C)_{\bar{\imath}i}\;.
        \end{align}
        \Cref{eq:differential_impacts} states that Type 1 events increase the active edge density more than Type 3 events decrease it. 
        This reflects a local asymmetry in the subcritical regime, between votes that increase the prevalence of a local minority opinion and votes that reduce it. 
        The asymmetry is due to the intervening rewiring steps, which tend to remove edges from the focal node $u$ prior to a Type 3 event.  
        Since Type 1 and Type 3 events occur at the same rate, our formalism predicts that voting events tend to increase the active-edge density when $\rho$ is small. 
        In \Cref{fig:symmetry_breaking}, we check this prediction by comparing the expressions in \Cref{eq:differential_impacts} to the distribution of all impacts $\Delta M_{01}$ on the active edge count due to voting events. 
        In the subcritical regime, the mean increment (black) is positive, reflecting the fact that Type 1 events (blue) outstrip Type 3 events (orange) in expected generation of active edges. 
        As $\rho$ grows, the separation-of-timescales assumption degrades, and the asymmetry between Type 1 and Type 3 events breaks down. 
        For large $\rho$, Type 1 and Type 3 events have similar increments in expectation and the distribution of $\Delta M_{01}$ becomes symmetric.

        Finally, we average over events of Types 1--3 to obtain the approximate expected increment in edge counts per voting event. 
        It is given by the four-vector
        \begin{align}
    		\hat{\mathbf{v}}(\q, \x) = \frac{1}{2} \sum_{i \in \{0,1\}} \frac{\E[E|i]\mathbf{e}_i(\C) + \E[F|i]\mathbf{f}(\C) + \E[G|i]\mathbf{g}_i(\mathbf{q},\mathbf{c})}{\E[E + F + G|i]}\;. \label{eq:v_hat}
    	\end{align}
    	For convenience, we summarize the expressions appearing in \Cref{eq:v_hat} in \Cref{tab:summary}. 
    	
    	   	\begin{table}
    	    \centering
    	    \begin{tabular}{l|l}
    	        Term &  Expression \\ 
    	        \hline 
    	         Type 1 expected increment & $\mathbf{e}_i(\C)_{01} = c_{\bar{\imath}i} - c_{ii} - 1$ \\ 
    	         Type 2 expected increment & $\mathbf{f}_i(\C)_{01} = \frac{\mathbf{e}_0(\C)_{01} + \mathbf{e}_1(\C)_{01}}{2}$ \\ 
    	         Type 3 expected increment & $\mathbf{g}_i(\C)_{01} = c_{\bar{\imath}i} + \varepsilon_i\frac{\beta_i}{1 -\beta_i}\left(1-\phi_{K_0}(\beta_i)\right)$ \\ 
    	         Type 1 expected count & $\E[E|i] = 1 - \phi_{K_0}(\beta_i)$ \\ 
    	         Type 2 expected count & $\E[F|i] = \sigma_i(c_{\bar{\imath}\bar{\imath}} + c_{\bar{\imath}i} - \mathbf{g}_{i}(\C)_{01})$\\ 
    	         Type 3 expected count & $\E[G|i] = 1 - \phi_{K_0}(\beta_i)$
    	    \end{tabular}
    	    \caption{Summary of the terms appearing in \Cref{eq:v_hat}. 
    	    Only the  $01$ components (corresponding to active edges) are shown.}
    	    \label{tab:summary}
    	\end{table}
    
    	Combining \Cref{eq:dynamics} with \Cref{eq:neighbor_vote,eq:wild_vote,eq:v_hat} and setting the lefthand side equal to zero gives our Markovian approximation for the arch: 
        \begin{align}
    		0 = \lambda \mathbf{w}(\x) + (1-\lambda)\alpha \mathbf{r}(\q) + (1-\lambda)(1-\alpha) \hat{\mathbf{v}}(\q, \x)\;. \label{eq:approx_dynamics}
    	\end{align}
    	This is a closed, deterministic equation in $\x$, derived under assumptions that are approximately correct near the critical point. 
    	Solving this equation yields $\hat{\x}$, the limit point of the approximate dynamics under \Cref{eq:approx_dynamics}.\footnote{In principle, \Cref{eq:approx_dynamics} may admit multiple limit points. Throughout our numerical experiments, we have found the limit point to be unique.}
    	The approximation indicates the subcritical case when $\hat{\rho}(\q;\alpha,\lambda) = 2\hat{\x}_{01}(\q;\alpha, \lambda) \leq 0$, and the supercritical case otherwise. 
    	Solving 
    	\begin{align}
    			\alpha^*(\q, \lambda) = \max \{\alpha : \hat{\rho}(\q; \alpha, \lambda) = 0\} \label{eq:transition_approx}
    	\end{align}
    	then gives our approximation for the critical value $\alpha^*$ at which persistent disagreement emerges.  
    	We again emphasize that this solution is independent of both the time step $t$ and the initialization of $\mathcal{G}$. 
    	
    	\Cref{fig:phase_transition_heatmap} compares numerical solutions to \Cref{eq:transition_approx} simulation data for the complete range of $q_1 \in [0,1]$. 
    	The accuracy of the approximation is strongest for $\q \approx \left(\frac{1}{2}, \frac{1}{2} \right)$ and on the rewire-to-random model variant. 
    	See also \Cref{fig:prev_results}(a) for comparisons of the solutions of \Cref{eq:transition_approx} to extant approximation schemes in the case $\q = \left(\frac{1}{2},\frac{1}{2}\right)$.
    	
    	\Cref{fig:phase_transition_heatmap} highlights one of the qualitative differences between the rewire-to-random and rewire-to-same model variants. 
    	As discussed in \Cref{sec:AVMs}, while  $\alpha^*$ depends strongly on $\q$ in the rewire-to-random model variant, it is independent of $\q$ in the rewire-to-same variant. 
    	This behavior is reflected algebraically in \Cref{eq:beta,eq:coefs}, which in turn govern the terms appearing in \Cref{eq:approx_dynamics}. 
    	The quantities $\beta$, $\varepsilon$, and $\sigma$ depend directly on $\q$ in the rewire-to-random model, regardless of the value of $\x$. 
    	However, in the rewire-to-same model, dependence on $\q$ emerges only when $\rho > 0$. 
    	This in turn implies that the phase transition is itself independent of $\q$, as is indeed observed in both the data and our approximation. 
    	Beyond the algebra, the localized approximation scheme we have developed gives to our knowledge the first mechanistic explanation of this difference in the phase transitions of the two models.\footnote{The pair-approximation (PA) equations of \cite{Durrett2012} predict this difference but the mechanism therein is less clear to us.} 
    	Consider the emergence of dissenting node $u$ with opinion $1$ on a component of majority opinion $0$. 
    	In the rewire-to-random model, the fast local rewiring dynamics depend explicitly on $\q$, the global opinion densities. 
    	When $q_1$ is large, an edge rewired away from $u$ is more likely to become inactive, resulting in fewer active edges in the neighborhood of $u$. 
    	This is in turn reflected by the term $g_1(\q,\C)_{01} = \frac{1}{2}\left(\E[K|i] - \E[J|i]\right)$ governing the impact of Type 3 events, whose magnitude enters into the calculation of the phase transition via \Cref{eq:approx_dynamics,eq:transition_approx}.  
    	In the rewire-to-same case, however, the fast rewiring does not explicitly depend on $\q$. 
    	An active edge attached to $u$ that rewires becomes inactive with probability $1$.
    	As a result, there is no dependence of Type 3 events on $\q$, and the phase transition is  independent of $\q$. 
    	
    	\begin{figure}
    	 	\centering
    	 	\includegraphics[width=\textwidth]{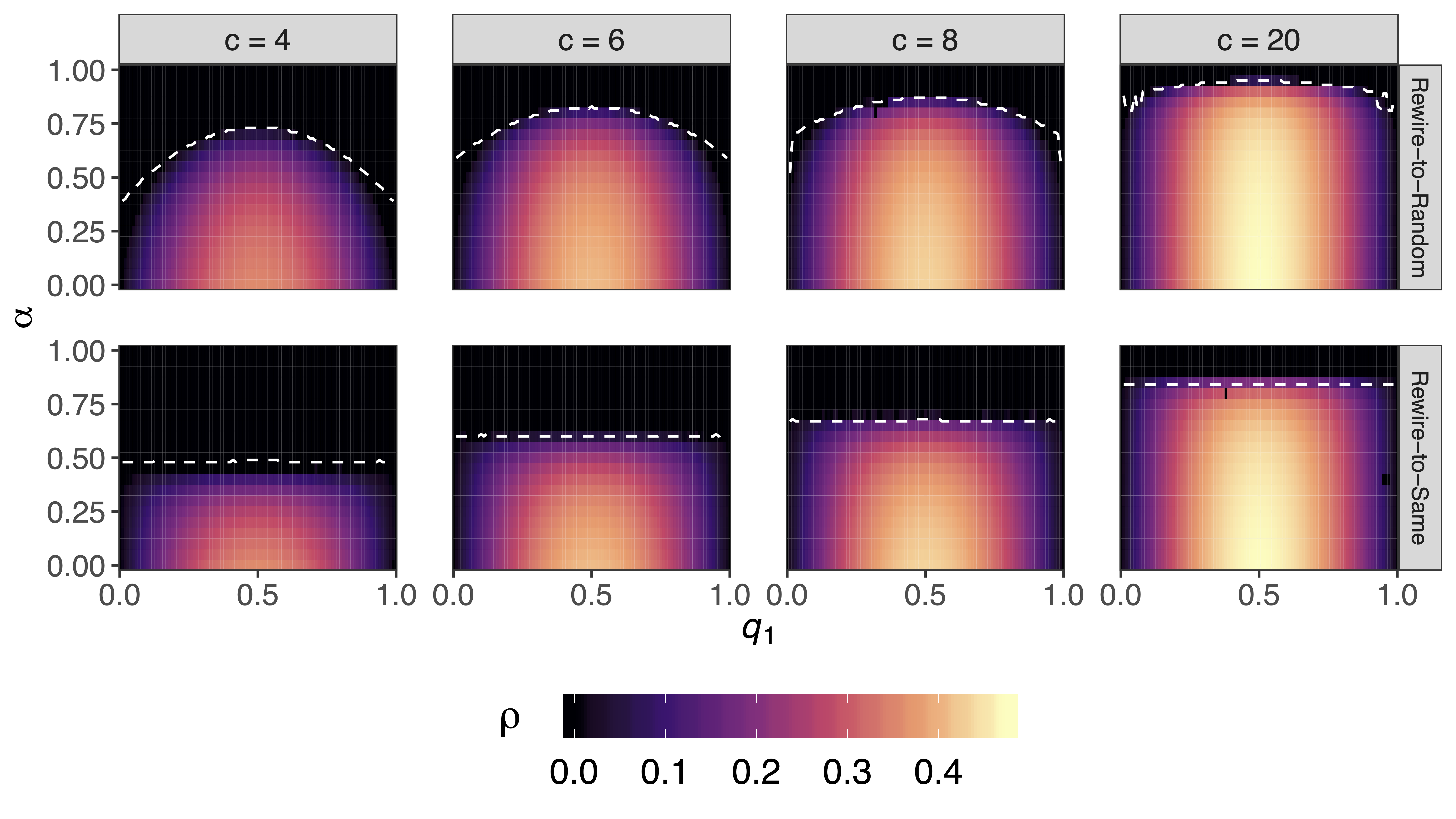}
    	 	\caption{Approximation of the phase transition $\alpha^*$ for rewire-to-random and rewire-to-same systems for varying mean degree $c$ and opinion densities $\q$ under the equilibrium measure $\eta$. 
    	 	Dashed lines give estimates of $\alpha^*(\q)$ obtained by numerically solving \Cref{eq:transition_approx}. 
    	 	Some numerical artifacts are visible in the rewire-to-random case for large $c$.
    	 	Color gives the equilibrium density of active edges $\rho$ from simulations carried out under the conditions described in \Cref{fig:prev_results}.
    	 	Both the estimate of $\alpha^*$ and the simulation results are independent of initialization of $\mathcal{G}$.} 
    	 	\label{fig:phase_transition_heatmap}
    	\end{figure} 
    	
    	We now turn to the approximation of $\rho^*(\q;\alpha, \lambda)$, the equilibrium density of active edges in the supercritical regime. 
    	In this regime, the distinction between local minority and majority nodes progressively erodes, as does the validity of the timescale-separation assumption. 
    	One way to view this erosion is in terms of decay of the impact of Type 3 events, as discussed in \Cref{fig:symmetry_breaking}.
    	As $\rho$ increases, the impact of a single Type 3 event progressively diminishes due to re-randomization of the focal node's local neighborhood. 
    	We model this re-randomization via a simple interpolation to a mean-field approximation of the arch in the case  $\alpha = 0$, which corresponds to a variant of the voter model without rewiring. 
    	We begin by deriving this approximation.  
    	
    	When $\alpha = \lambda = 0$, active edges enter and exit the system only through voting events. 
        We have already written the mean-field approximation for the expected impact of a voting event in \Cref{eq:neighbor_vote}. 
        When only these events take place, the equilibrium condition is $e_i(\C) = 0$ for $i = 0,1$.
        It suffices to solve the system
        \begin{align*}
        	0 &=  1 + c_{10} - c_{00} \\ 
        	0 &=  1 + c_{01} - c_{11} 
        \end{align*}
        for $\C$ and subsequently for $\x$.  
        We recall that $c_{ij} = c{x_{ij}}/{q_i}$ and that $2x_{01} = 1 - x_{00} - x_{11}$. 
        Substituting these relations we obtain
        \begin{align*}
        \frac{2q_0q_1}{c}\left(\begin{matrix}1 \\ 1\end{matrix}\right) + \q = \left[\begin{matrix} 1 + q_1 & q_0 \\ q_1 & 1 + q_0 \end{matrix}\right]\left(\begin{matrix}x_{00} \\ x_{11}\end{matrix}\right)\;.
        \end{align*}
        The unique solution is 
        \begin{align*}
        \left(\begin{matrix}x^*_{00} \\ x^*_{11}\end{matrix}\right) = \frac{q_0 q_1}{c} \mathbf{e} + \frac{1}{2}\left(\begin{matrix} q_0(1 + q_0 - q_1) \\ q_1(1 + q_1 - q_0) \end{matrix}\right)\;.
        \end{align*}
        We may then compute the mean-field approximation for the $\alpha = 0$ arch: 
        \begin{align*}
        \hat{\rho}^*(\q) = 2x_{01} 
             = 1 - x^*_{00} - x^*_{11} 
             = 2q_0 q_1 \frac{c-1}{c}.
        \end{align*}
        We note that this result is identical to that derived in \cite{vazquez2008analytical} for a node-updating non-adaptive voter model.
        
        We now introduce the interpolation function 
    	\begin{align}
    		s(\q,\x) = \frac{\hat{\rho}^*(\q) - \rho}{\hat{\rho}^*(\q)} \label{eq:interpolation}
    	\end{align}
    	to quantify the distance of the system state from the estimated $\alpha = 0$ arch.
    	We then use this interpolation function to introduce decay in Type 3 events,  replacing $\mathbf{g}(\q, \C)$ in \Cref{eq:v_hat} with $\tilde{\mathbf{g}}(\q, \C) = \mathbf{g}(\q, \C)s(\q,\x)$.
    	The corresponding solution for $\hat{\x}$  yields the supercritical approximation of $\rho$. 
    	
    	\Cref{fig:arch_approximations} shows the resulting approximations for the arch in both models, across a range of parameter regimes and both model variants.
    	The arches for the rewire-to-random model agree well with the data on both the support of the arch and the equilibrium active edge density. 
    	The rewire-to-same arches are somewhat less precise. 
    	The arches do correctly span the complete interval $[0,1]$. 
    	The overall numerical agreement with the data is comparable to extant methods, but the parabolic shape of the arch is not completely reproduced --- there is some warping near the base. 
    	The reason for this warping is not clear to us at present, and further investigation into this phenomenon may yield  both theoretical and computational progress. 
    	\begin{figure}
    	 	\centering
    	 	\includegraphics[width=\textwidth]{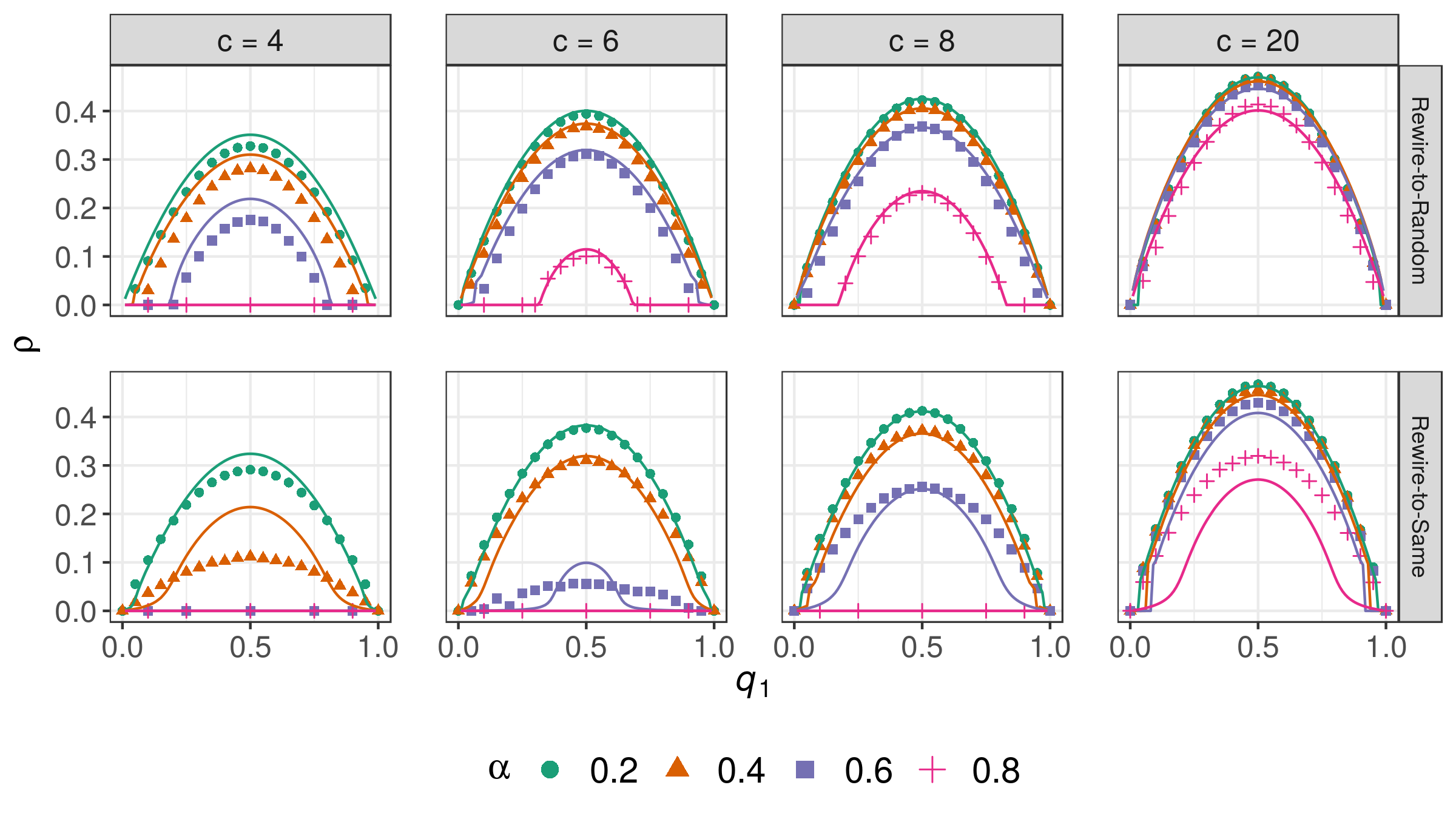}
    	 	\caption{Approximations to the arch for varying $\alpha$, $\q$, and $c$. 
    	 	Points give averages over simulation runs on AVMs under the conditions described in \Cref{fig:prev_results}.
    	 	Solid lines give the equilibrium value of $\hat{\rho}$ obtained by numerically solving for the fixed points of \Cref{eq:dynamics} using the interpolation function given by \Cref{eq:interpolation}. 
    	 	 Higher arches correspond to lower values of $\alpha$.} 
    	 	\label{fig:arch_approximations}
    	\end{figure} 
    	
\section{Discussion} \label{sec:discussion}
    The Markovian approximation technique we have developed offers predictions for the equilibrium active edge density $\rho^*$ across the entirety of parameter space, and for varying opinion densities $\q$. 
    Its accuracy in these tasks is generally comparable to that of the best extant methods. 
    For example, \Cref{fig:prev_results}(a) shows that our Markovian approximation is at least as accurate as AMEs \cite{Durrett2012} in predicting the $c = 4$ phase transition for the rewire-to-random model, and grows more accurate as $c$ increases. 
    Our approximation is substantially more accurate than AMEs for the rewire-to-same phase transition, and only slightly less accurate than the compartmental approach of \cite{Bohme2011} for this model variant. 
    Relatively few approximation schemes make predictions for the full arch, and it is therefore more difficult to make quantitative comparisons. 
    The compartmental method of \cite{Demirel2012} approximates the equilibrium active edge density at $\q = \left(\frac{1}{2},\frac{1}{2}\right)$ in the rewire-to-same variant more accurately than our method (\Cref{fig:prev_results}(b)), but does not make predictions for asymmetric opinion densities. 
    AME predictions \cite{Durrett2012} recover the asymmetric phase transition and arches reasonably well in the rewire-to-random case, but are much less accurate for the rewire-to-same variant. 
    Whereas the AME arches display warping in the rewire-to-random variant, our Markovian approximation displays warping in the rewire-to-same variant. 
    In the $c = 4$ case shown in \Cref{fig:prev_results}(b), the present method offers overall accuracy in computing the arches similar to that of the AMEs, and improves as $c$ increases. 
    
    \subsection{Computational Considerations} 
    
        Solving \cref{eq:approx_dynamics} requires finding the solution of a system of four coupled nonlinear equations, which may be done efficiently using a standard numerical solver. 
        Notably, the dimensionality of the approximation is independent of the mean degree $c$. 
        This contrasts to compartmental methods \cite{Durrett2012,Bohme2011,Silk2014,Demirel2012}, the dimension of which generally display quadratic or higher scaling in $c$.
        For example, AMEs comprise a system of  $\Theta\left(k_\mathrm{max}^2\right)$ coupled differential equations, where $k_\mathrm{max}$ is the highest node degree expected to be encountered in simulation; in the $c = 4$ case, the authors of \cite{Durrett2012} used $272$ such equations.
        This scaling makes the computation of approximations computationally prohibitive even for modest mean degree $c$. 
        Similarly, the method of \cite{Bohme2011} for approximating the rewire-to-same phase transition requires a bisection search in $\alpha$ for which each function evaluation corresponds to finding the largest eigenvalue of a $(c-1) \times (c-1)$ matrix.
        The scaling is thus at least $O\left((\log \frac{1}{\epsilon})(c-1)^2\right)$, where $\epsilon$ is the desired approximation accuracy. 
        Because our proposed method scales independently of $c$, it can be used to approximate AVMs with arbitrarily large mean degrees.

	\subsection{Conclusions}
	
    	Adaptive voter models offer a simple set of mechanisms that generate emergent opinion-based assortativity in complex networks. 
    	While the underlying rules are simple to state, the coevolving nature of the dynamics render these systems interesting and challenging to analyze. 
    	We have considered an ergodic adaptive voter model variant which enables a local perspective on fragmentation transitions and other model properties. 
    	The local perspective allows us to use the asymmetry of voting events to develop Markovian approximations based on the fast timescale dynamics around single nodes. 
    	The resulting approach is conceptually intuitive, computationally tractable, and predictively performant. 
    	
    	One of the most puzzling issues raised by our results is the difference between the accuracies of our approach for the rewire-to-random and rewire-to-same adaptive voter models. 
    	While we succeed in characterizing the rewire-to-random arch nearly exactly, the same methods produce poorer results for the rewire-to-same model. 
    	We conjecture that the rapid local sorting produced in rewire-to-same dynamics violates our mean-field assumption on Type 1 events, which would lead to approximation degradation. 
    	It would be interesting to extend our methodology to see whether refinements are possible that better characterize the rewire-to-same behavior and shed further light on the essential features governing the dramatic difference in the nature of the phase transitions in these two models. 
    		
    	It is also of interest to consider extensions and generalizations. 
    	The most natural extension is to the case of multiple opinion states and structured opinion spaces. 
    	Previous work on multi-opinion models has been restricted to either approximation of the various phase transitions \cite{Bohme2012} or empirical discussion of supercritical equilibrium behavior \cite{Shi2013}. 
    	One reason for this is computational. 
    	The number of operations required to compute approximations under active-motif and AME approaches are exponential in the number $\abs{\mathcal{X}}$ of opinion states, rendering both methods infeasible. 
    	In contrast, likely extensions of our local Markovian approximation methods scale as $\Theta(\abs{\mathcal{X}}^2)$, which would offer a significant reduction in computing time. 
    	If accuracy were preserved, such extensions would present the first scalable analytic methods for multi-opinion models. 
    	Other generalizations are also possible.  
    	While we developed our approximations for the specific case of the binary-state AVM, that development relies only on ergodicity, timescale separation, and the mean-field assumption. 
    	We conjecture that these ingredients should be present in any adaptive model with homophilic dynamics in which rewiring steps involve uniform selection from an extensive subset of the graph, such as a subset sharing a given node state. 
    	An example of a more complex system in which these ingredients are present is the networked evolutionary prisoner's dilemma game of \cite{Lee2018}, in which nodes display richer strategic behavior in their opinion update and rewiring behavior. 
    	The existence of a phase transition driven by homophily may allow for the deployment of our novel methods in such cases as well. 

\section*{Acknowledgments} 
	We are grateful to Feng (Bill) Shi for contributing code used for simulations, Hsuan-Wei Lee for contributing code used to construct the approximate master equation solutions shown in \Cref{fig:prev_results}, and Patrick Jaillet for helpful discussions. 

\section*{Code}

    Documented code for running simulations and computing the approximations described in this paper may be found at \url{https://github.com/PhilChodrow/AVM}.

\bibliographystyle{siamplain}
\bibliography{refs.bib}{}

\end{document}